\documentclass[letterpaper,10pt]{article}

\usepackage{makeidx}
\usepackage{amsmath, amsfonts, amssymb, amsthm}
\usepackage{url, paralist,enumitem, multirow, tabularx}
\usepackage[vlined,ruled,boxed]{algorithm2e}

\usepackage[svgnames]{xcolor}

\definecolor{darkgreen}{rgb}{0,.35,0}
\definecolor{darkblue}{rgb}{0,0,.35}
\definecolor{darkred}{rgb}{.35,0,0}
\usepackage[pdfauthor={Andrew Arnold, Mark Giesbrecht and Daniel S. Roche},
            pdftitle={Faster sparse interpolation of straight-line programs},colorlinks,linkcolor=DarkBlue,
            citecolor=DarkGreen,urlcolor=darkred]{hyperref}

\usepackage[round]{natbib}
\SetKw{by}{by}
\SetKw{downto}{down to}
\SetKwRepeat{DoWhile}{do}{while}

\newtheorem{Thm}{Theorem}
\newtheorem{Cor}[Thm]{Corollary}
\newtheorem{Lem}[Thm]{Lemma}

\newtheorem{Defn}[Thm]{Definition}

\newtheorem{Exam}[Thm]{Example}
\newtheorem{Fact}[Thm]{Fact}

\renewcommand{\O}{{\mathcal{O}}}
\newcommand{\COL}[1]{\mathcal{C}_g(#1)}
\newcommand{\softO}{{\O\mskip2mu\tilde{\,}\mskip1mu}}
\newcommand{\R}{{\mathcal{R}}}

\renewcommand{\th}{^{\text{th}}}

\newcommand{\fhat}{{\hat f}}
\newcommand{\M}{{\mathsf{M}}}
\newcommand{\ZZ}{{\mathbb{Z}}}
\newcommand{\CC}{{\mathbb{C}}}
\newcommand{\RR}{{\mathbb{R}}}
\newcommand{\Q}{{\mathcal{Q}}}
\newcommand{\SLP}{{\mathcal{S}}}

\SetProcNameSty{texttt}
\SetFuncSty{texttt}
\SetTitleSty{texttt}{}
\title{Faster sparse interpolation of straight-line programs%
\thanks{A version of this paper appeared at CASC 2013,
\href{http://dx.doi.org/10.1007/978-3-319-02297-0_5}{\texttt{doi:10.1007/978-3-319-02297-0\_5}}}}

\usepackage{authblk}

\author[1]{Andrew Arnold}
\author[1]{Mark Giesbrecht}
\affil[1]{Cheriton School of Computer Science\\
University of Waterloo}
\author[2]{Daniel S.\ Roche\thanks{Supported by NSF Award \#1319994}}
\affil[2]{Computer Science Department\\
United States Naval Academy}

\begin{document}

\maketitle

\begin{abstract}
  We give a new probabilistic algorithm for interpolating a ``sparse''
  polynomial $f$ given by a straight-line program.  Our algorithm
  constructs an approximation $f^*$ of $f$, such that $f-f^*$
  probably has at most half the number of terms of $f$, then
  recurses on the difference $f-f^*$.  Our approach builds on previous
  work by \cite{GarSch09}, and \cite{GieRoc11}, and is asymptotically
  more efficient in terms of the total cost of the probes required
  than previous methods, in many cases.
\end{abstract}

\section{Introduction}
\label{sec:intro}

We consider the problem of interpolating a sparse, univariate
polynomial 
\[
f = c_1 z^{e_1} + c_2 z^{e_2} + \cdots + c_t z^{e_t} \in\R[z]
\]
of degree $d$ with $t$ non-zero coefficients $c_1,\ldots,c_t$ (where
$t$ is called the \emph{sparsity} of $f$) over a ring
$\R$.  More formally, we are given a \emph{straight-line program} that
evaluates $f$ at any point, as well as bounds $D \geq d$ and $T \geq
t$.  The straight-line program is a simple but useful abstraction of a
computer program without branches, but our interpolation algorithm
will work in more common settings of ``black box'' sampling of $f$.

We summarize our final result as follows.

\begin{Thm}
  \label{thm:cost}
  Let $f \in \R[z]$, where $\R$ is any ring.  Given any straight-line
  program of length $L$ that computes $f$, and bounds $T$ and $D$ for
  the sparsity and degree of $f$, one can find all coefficients and
  exponents of $f$ using $\softO(L T\log^3 D + LT\log D\log
  (1/\mu))$\footnotemark[4] ring operations in $\R$, plus a similar
  number of bit operations.  The algorithm is probabilistic of the
  Monte Carlo type: it can generate random bits at unit cost and on
  any invocation returns the correct answer with probability greater
  than $1-\mu$, for a user-supplied tolerance $\mu>0$.
\end{Thm}

\footnotetext[4]{For summary convenience we use soft-Oh notation: for
  functions $\phi,\psi\in\RR_{>0}\to\RR_{>0}$ we say
  $\phi\in\softO(\psi)$ if and only if $\phi\in
  \O(\psi(\log\psi)^c)$ for some constant $c\geq 0$.}


\subsection{The straight-line program model and interpolation}
\label{ssec:slpinterp}

Straight-line programs are a useful model of computation, both as a
theoretical construct and from a more practical point of view; see,
e.g., \citep[Chapter 4]{BCS97}.  Our interpolation algorithms work more
generally for $N$-variate sparse polynomials $f\in\R[z_1,\ldots,z_N]$
given by a straight-line program $\SLP_f$ defined as follows.
$\SLP_f$ takes an input $(a_1, \dots, a_N) \in
\R^N$ of length $N$, and produces a vector $b \in \R^L$ via a series
of $L$ instructions $\Gamma_i : 1 \leq i \leq L$ of the form
\[ 
\Gamma_i = 
\begin{cases}
  \gamma_i & \longleftarrow \alpha_1 \star \alpha_2, ~~\emph{or} \\
  \gamma_i & \longleftarrow \delta \in \R ~~~~\emph{(i.e., a constant
    from $\R$)},
\end{cases} 
\]
where $\star$ is a ring operation `$+$', `$-$', or `$\times$',
and either $\alpha_\ell \in \{ a_j \}_{1 \leq j \leq n}$ or 
$\alpha_\ell \in \{ \gamma_k \}_{1 \leq k < i}$ for $\ell=1,2$. 
When we say $\SLP_f$ \emph{computes}
$f$, we mean $\SLP_f$ sets $\gamma_L$ to $f( a_1, \dots, a_N)\in\R$.

To interpolate an $N$-variate polynomial $f\in\R[z_1,\ldots,z_N]$, we
apply a Kronecker substitution, and interpolate
\[
\fhat(z) = f\left(z, z^{(D+1)}, z^{(D+1)^2}, \dots,
  z^{(D+1)^{N-1}} \right)\in\R[z].
\]
While this certainly increases the degree, $f$ and $\fhat$ have the
same number of non-zero terms, and $f$ can be easily recovered from
$\fhat$.  This reduces the problem of interpolating the $N$-variate
polynomial $f$ of partial degree at most $D$ to interpolating a
univariate polynomial $\fhat$ of degree at most $(D+1)^N$. For the
remainder of this paper we thus assume $f$ is univariate.

It will also be necessary to evaluate our polynomial $f\in\R[z]$, or
rather our straight-line program $\SLP_f$ for $f$, in an extension ring of
$\R$.  Precisely, we want to evaluate $f$ at symbolic $\ell$th roots
of unity for various choices of $\ell$, or algebraically, in
$\R[z]/(z^{\ell}-1)$.  This may be regarded as transforming our
straight-line program by substituting operations in $\R$ with
operations in $\R[z]/(z^\ell-1)$, where each element is represented by
a polynomial in $\R[z]$ of degree less than $\ell$. Each instruction
$\Upsilon_i$ in the transformed branching program now potentially
requires $\M( \ell)$ operations in $\R$, where $\M(\ell)$ is the
number of operations in $\R$ and bit operations needed to multiply two
degree-$\ell$ polynomials over the base ring $\R$. By \cite{CanKal91},
we may assume $\M(\ell) = \O(\ell \log \ell \log\log \ell)$.

Each evaluation of our straight-line program for $f$ in
$\R[z]/(z^\ell-1)$ is called a \emph{probe of degree $\ell$}.  Thus,
the cost of a degree-$\ell$ probe to $\SLP_f$ is $\softO( L \ell )$ operations
in $\R$, and similarly many bit operations.

This is easily connected to the more ``classical'' view of sparse
interpolation, in which probes are simply evaluations of a ``black-box'' 
polynomial at a single point (and we do not have any
representation for how $f$ is calculated). Each probe in the
straight-line program model can
be thought of as evaluating $f$ at \emph{all} $\ell$th roots of unity
in the classical model.  Since we charge $\M(\ell)=\softO(\ell)$
operations in $\R$ for a degree $\ell$ probe in the straight-line
program model, i.e., about $\ell$ times
as much as a single black-box probe, this is consistent with the costs in a
classical model.  We note that algorithms for sparse interpolation
presented below could be stated in this classical model, though we
find the straight-line program model convenient and will continue with
it throughout this paper.


\subsection{Previous work}
\label{ssec:prev}

Straight-line programs, or equivalently algebraic circuits,
are important both as a computational model and as a data structure
for polynomial computation. Their rich history includes both
algorithmic advances and practical implementations
\citep{Kal89,StuZha90,BruHei02}.

One can naively interpolate a polynomial $f\in\R[z]$ given by a straight-line 
program
using a dense method, with $D$ probes of degree $1$. Prony's
\citeyearpar{Prony95} interpolation algorithm --- see \citep{BenTiw88,
  KalLak90, GieLab09} --- is a sparse interpolation method that uses
evaluations at only $2T$ powers of a root of unity whose order is
greater than $D$. However, in the straight-line program model for a
general ring, this would require evaluating at a symbolic $D$th root
of unity, which would use at least $\Omega(D)$ ring operations and
defeat the benefit of sparsity.  Problems with Prony's algorithm are
also seen in the classical model in that the underlying base ring $\R$
must also support an efficient discrete logarithm algorithm on entries
of high multiplicative order (which, for example, is not feasible
over large finite fields).

We mention two algorithms specifically intended for straight-line
programs.

\subsubsection{The Garg-Schost deterministic algorithm.}

\cite{GarSch09} describe a novel deterministic algorithm for interpolating a
multivariate polynomial $f$ given by a straight-line program.  Their
algorithm entails constructing an integer symmetric polynomial with
roots at the exponents of $f$:
\[
  \chi = \prod_{i=1}^{t}(y-e_i) \in\ZZ[y],
\]
which is then factored to obtain the exponents $e_i$.

Their algorithm first finds a \emph{good} prime: a prime $p$ for which
the terms of $f$ remain distinct when reduced modulo $z^p-1$. We call
such an image $f \bmod (z^p-1)$ a \emph{good image}. Such an image
gives us the values $e_i \bmod p$ and hence $\chi(y) \bmod p$.

\begin{Exam}
  For $f = z^{33}+z^3$, $5$ is not a good prime because $f \bmod
  (z^5-1) = 2z^3$.  We say $z^{33}$ and $z^{3}$ {\em collide} modulo
  $z^5-1$.  $7$ is a good prime, as the image $f(z) \bmod (z^7-1) =
  z^5+z^3$ has as many terms as $f(z)$ does.
\end{Exam}

In order to guarantee that we have a good prime, the algorithm
requires that we construct the images $f \bmod (z^p-1)$ for the first
$N$ primes, where $N$ is roughly $\softO( T^2\log D)$. A good prime
will be a prime $p$ for which the image $f \bmod (z^p-1)$ has
maximally many terms, which will be exactly $t$. Once we know we have
a good image we can discard the images $f \bmod (z^q-1)$ for
\emph{bad} primes $q$, i.e. images with fewer than $t$ terms. We use
the remaining images to construct $\chi(y) = \prod_{i=1}^{t}(y-e_i)
\in \mathbb{Z}[y]$ by way of Chinese remaindering on the images
$\chi(y) \bmod p$.

We factor $\chi(y)$ to obtain the exponents $e_i$, after which we
directly obtain the corresponding coefficients $c_i$ directly from a
good image.

The algorithm of \cite{GarSch09} can be made faster, albeit Monte Carlo,
using the following number-theoretic fact.

\begin{Fact}[\citealp{GieRoc11}]
  Let $f \in \mathcal{R}[z]$ be a polynomial with at most $T$ terms
  and degree at most $D$.  Let $\lambda = \max(21, \lceil
  \tfrac{5}{3}T(T-1)\log D \rceil )$.  A prime $p$ chosen at random in
  the range $[ \lambda, 2\lambda]$ is a good prime for $f(z)$ with
  probability at least $\tfrac{1}{2}$.
\end{Fact}

Thus, in order to find a good image with probability at least
$1-\varepsilon$, we can inspect images $f \bmod (z^p-1)$ for $\lceil
\log 1/ \varepsilon \rceil$ primes $p$ chosen at random in $[ \lambda,
2\lambda ]$. As the height of $\chi(y)$ can be roughly as large as
$D^T$, we still require some $\mathcal{O}^{\sim}(T \log D)$ probes to
construct $\chi(y)$.

\subsubsection{The ``diversified'' interpolation algorithm.}

\cite{GieRoc11} obtain better performance by way of
\emph{diversification}. A polynomial $f$ is said to be \emph{diverse}
if its coefficients $c_i$ are pairwise distinct. The authors show
that, for $f$ over a finite field or $\CC$ and for appropriate random
choices of $\alpha$, $f(\alpha z)$ is diverse with probability at
least $\tfrac{1}{2}$. They then try to interpolate the diversified
polynomial $f(\alpha z)$.

Once we have $t$ with high probability, we look at images $f(\alpha z)
\bmod (z^p-1)$ for primes $p$ in $[\lambda, 2\lambda]$, discarding bad
images. As $f(\alpha z)$ is diverse, we can recognize which terms in
different good images are images of the same term. Thus, as all the
$e_i$ are at most $D$, we can get all the exponents $e_i$ by looking
at some $\softO(\log D)$ good images of $f$.

\subsection{Deterministic zero testing}

Both the Monte Carlo algorithms of \cite{GarSch09} and \cite{GieRoc11}
can be made Las Vegas (i.e., no possibility of erroneous output, but unbounded 
worst-case running time) by way of
deterministic zero-testing. Given a polynomial $f$ represented by a
straight-line program, each of these algorithms produces a polynomial
$f^*$ that is probably $f$.

\begin{Fact}[\cite{BlaHar09}; Lemma 13]
  Let $\R$ be an integral domain, and suppose $f = f^* \bmod (z^p-1)$ for
  $T\log D$ primes.  Then $f = f^*$.
\end{Fact}

Thus, testing the correctness of the output of a Monte Carlo algorithm
requires some $\softO( T\log D )$ probes of degree at most
$\softO(T\log D)$. This cost does not dominate the cost of either
Monte Carlo algorithm. We note that this deterministic zero test can
dominate the cost of the interpolation algorithm presented in this
paper if $T$ is asymptotically dominated by $\log D$.


\subsection{Summary of results}

We state as a theorem the number and degree of probes required by our
new algorithm presented in this paper.

\begin{Thm}\label{thm:probe}
  Let $f \in \R[z]$, where $\R$ is a ring.  Given a straight-line
  program for $f$, one can find all coefficients and exponents of $f$
  with probability at least $1-\mu$ using $\softO\left( \log T(\log D
    + \log \tfrac{1}{\mu}) \right)$ probes of degree at most $\O( T
  \log^2 D )$.
\end{Thm}

\begin{table}
  \caption{A ``soft-Oh'' comparison of interpolation algorithms for straight-line programs} 
  \label{tab:comp}
  \setlength{\tabcolsep}{5pt}
  \centering
  \begin{tabular}{l | r|r|r|r}
    & Probes & Probe degree & Cost of probes & Type \\
    \hline
    Dense & $D$ & $1$ & $LD$ & deterministic \\\hline
    Garg \& Schost & $T^2\log D$ & $T^2\log D$ & $LT^3\log^2 D$ & deterministic \\\hline
    *Las Vegas G \& S & $T\log D$ & $T^2\log D$ & $LT^3\log^2 D$ & Las Vegas \\\hline
    *Diversified & $\log D$ & $T^2\log D$ & $LT^2\log^2 D$ & Las Vegas \\\hline
    $\dagger$Recursive & $\log T\log D$ & $T \log^2 D$ & $LT\log^3 D$ & Monte Carlo \\\hline
    \multicolumn{5}{l}{ \rule{0pt}{12pt}\footnotesize 
      *Average \# of probes given; $\dagger$ for a fixed probability of failure $\mu$ }
\end{tabular}
\end{table}

Table \ref{tab:comp} gives a rough comparison of known algorithms.
Our recursive algorithm improves by a factor of $T/\log D$
over the Giesbrecht-Roche diversification algorithm --- ignoring
``soft'' multiplicative factors of $(\log(T/\log D))^{O(1)}$ --- and as such is
better suited for moderate values of $T$.  Our algorithm recursively interpolates
a series of polynomials of decreasing sparsity.  An advantage of this
method is that, when we cross a threshold where $\log D$ begins to dominate
$T$, we can merely call the Monte Carlo diversification algorithm instead.


\section{A recursive algorithm for interpolating $f$}

Entering each recursive step in our algorithm we have our polynomial
$f$ represented by a straight-line program, and an explicit sparse
polynomial $f^*$ ``approximating'' $f$, that is, whose terms mostly appear in 
the sparse representation of $f$.
At each recursive step we try to interpolate the difference $g=f-f^*$.
To begin with, $f^*$ is initialized to zero.

We first find an ``ok'' prime $p$ which separates most of the terms of
$g$.  We then use that prime $p$ to build a approximation $f^{**}$,
containing most of the terms of $g$, plus possibly some additional
``deceptive'' terms.  The polynomial $f^{**}$ is constructed such that
$g=f-f^*$ has, with high probability, at most $T/2$ terms.  We then
recursively interpolate the difference $g-f^{**}$.  

Producing images $f^* \bmod (z^\ell-1)$ is straightforward, we merely
reduce the exponents of terms of $f^*$ modulo $\ell$.  We assume $g$
has a sparsity bound $T_g \leq T$.

\subsection{A weaker notion of ``good" primes}\label{sec:ok}

To interpolate a polynomial $g$, the sparse interpolation algorithm
described by \cite{GieRoc11} requires a \emph{good} prime $p$ which
keeps the exponents of $g$ distinct modulo $p$.  That is, $g \bmod
(z^p-1)$ has the same number of terms as $g$.  We define a weaker notion
of a good prime, an \emph{ok prime}, which separates most of the
terms of $g$.  To that end we measure, for fixed $g$ and prime $p$,
how well $p$ separates the terms of $g$.

\begin{Defn}
  Fix a polynomial $g = \sum_{i=1}^t c_iz^{e_i}\in\R[z]$ with non-zero
  $c_1,\ldots,c_t\in\R$, where $e_i < e_j$ for $i<j$, we say
  $c_iz^{e_i}$ and $c_jz^{e_j}$, $i \neq j$, \emph{collide} modulo
  $z^p-1$ if $e_i\equiv e_j \bmod p$.  We call any term $c_iz^{e_i}$ of $f$
  which collides with any other term of $f$ a \emph{colliding term} of $f$ 
  modulo
  $z^p-1$.  We let $\COL{p} \in [0, t]$ denote the number of
  colliding terms of $g$ modulo $z^p-1$.
\end{Defn}

\begin{Exam}
  For the polynomial $g= 1+z^5+z^7+z^{10}$, $\COL{2} = 4$, since $1$
  collides with $z^{10}$ and $z^5$ collides with $z^7$ modulo $z^2-1$.
  Similarly, $\COL{5} = 2$, since $z^5$ collides with $z^{10}$ modulo
  $z^5-1$.
\end{Exam}

We say $c_iz^{e_i}$ and $c_jz^{e_j}$ collide modulo $z^p-1$ because
both terms have the same exponent once reduced modulo $z^p-1$.  All
other terms of $g$ we will call {\em non-colliding} terms modulo
$z^p-1$.

In the sparse interpolation algorithm of \cite{GieRoc11}, one chooses
a $\lambda \in \mathbb{Z}_{>0}$ such that the probability of a prime
$p \in [ \lambda, 2\lambda ]$, chosen at random and having $\COL{p}=0$,
is at least $\tfrac{1}{2}$.  However, in order to guarantee that we
find such a prime with high probability, we need to choose $\lambda
\in \O( T^2\log D)$.

In this paper we will search over a range of smaller primes, while
allowing for a reasonable number of collisions.  We try to pick
$\lambda$ such that
\begin{equation*}
  \Pr\left( \COL{p} \geq \gamma \text{ for a random prime } p 
    \in [\lambda, 2\lambda] \right) < 1/2,
\end{equation*}
for a parameter $\gamma$ to be determined.

\begin{Lem}\label{lem:lambda}
  Let $g \in \R[z]$ be a polynomial with $t \leq T$ terms and degree
  at most $d \leq D$.  Suppose we are given $T$ and $D$, and let
  $\lambda = \max\left(21, \left\lceil \tfrac{10T(T-1)\ln(D)}{3
        \gamma}\right\rceil\right)$.  Let $p$ be a prime chosen at
  random in the range $\lambda, \dots, 2\lambda$.  Then $\COL{p} \geq
  \gamma$ with probability less than $\tfrac{1}{2}$.
\end{Lem}

\begin{proof}
  The proof follows similarly to the proof of Lemma 2.1 in
  \citep{GieRoc11}.

  Let $B$ be the set of unfavourable primes for which $\COL{p} \geq
  \gamma$ terms collide modulo $z^p-1$, and denote the size of $B$ by
  $\#B$.  As every colliding term collides with at least one other
  term modulo $z^p-1$, we know $p^{\COL{p}}$ divides $\prod_{1 \leq i
    \neq j \leq t}(e_i - e_j)$.  Thus, as $\COL{p} \geq \gamma$ for
  $p \in B$,

  \[
  \lambda^{\#B \gamma} \leq \prod_{p \in B}p^\gamma \leq \prod_{1 \leq
    i \neq j < t}(e_i - e_j) \leq d^{t(t-1)} \leq D^{T(T-1)}.
  \] 
  Solving the inequality for $\#B$ gives us
  \[
  \#B \leq \frac{ T(T-1)\ln(D) }{ \ln( \lambda )\gamma}.
  \]
  The total number of primes in $[ \lambda, 2\lambda]$ is greater than
  $3\lambda/(5\ln(\lambda))$ for $\lambda \geq 21$ by Corollary 3 to
  Theorem 2 of \citep{RosSch62}.  From our definition of $\lambda$ we have
  \[
  \frac{3\lambda}{5\ln(\lambda)} > \frac{2T(T-1)\ln(D)}{\ln(\lambda)\gamma} \geq 2\# B,
  \]
  completing the proof. \qed
\end{proof}

\subsubsection{Relating the sparsity of $g \bmod (z^p-1)$ with
  $\COL{p}$}
\quad\smallskip

Suppose we choose $\lambda$ according to Lemma \ref{lem:lambda}, and
make $k$ probes to compute $g \bmod (z^{p_1}-1), \dots, g \bmod
(z^{p_k}-1)$.  One of the primes $p_i$ will yield an image with fewer
than $\gamma$ colliding terms (i.e. $\COL{p_i} < \gamma$) with
probability at least $1-2^{-k}$.  Unfortunately, we do not know which
prime $p$ maximizes $\COL{p}$.  A good heuristic might be to select
the prime $p$ for which $g \bmod (z^{p}-1)$ has maximally many terms.
However, this does not necessarily minimize $\COL{p}$.  Consider the
following example.

\begin{Exam}
  Let
  \begin{equation*}
    g = 1 + z + z^4 - 2z^{13}.
  \end{equation*}
  We have
  \[
  g \bmod (z^2-1) = 2 - z, ~~\mbox{and}~~~ g \bmod (z^3-1) = 1.
  \]
  While $g \bmod (z^2-1)$ has more terms than $g \bmod (z^3-1)$,
  we see that $\COL{2} = 4$ is larger than $\COL{3} = 3$.
\end{Exam}

While we cannot determine the prime $p$ for which $g \bmod (z^p-1)$
has maximally many non-colliding terms, we show that choosing the
prime $p$ which maximizes the number of terms in $g \bmod (z^p-1)$ is,
in fact, a reasonable strategy.

We would like to find a precise relationship between $\COL{p}$, the
number of terms of $g$ that collide in the image $g \bmod (z^p-1)$,
and the sparsity $s$ of $g \bmod (z^p-1)$.

\begin{Lem}\label{lem:chooseMostTerms}
  Suppose that $g$ has $t$ terms, and $g \bmod (z^p-1)$ has $s \leq t$
  terms.  Then $t-s \leq \COL{p} \leq 2(t-s)$.
\end{Lem}
\begin{proof}
  To prove the lower bound, note that $t-\COL{p}$ terms of $g$ will
  not collide modulo $z^p-1$, and so $g \bmod (z^p-1)$ has sparsity $s$ at
  least $t-\COL{p}$.

  We now prove the upper bound.  Towards a contradiction, suppose that
  $\COL{p} > 2(t-s)$.  There are $\COL{p}$ terms of $g$ that collide modulo
  $z^p-1$.  Let $h$ be the $\COL{p}$-sparse polynomial
  comprised of those terms of $g$.  
  As each term of $h$ collides with at
  least one other term of $h$, $h \bmod (z^p-1)$ has sparsity at most
  $\COL{p}/2$.  Since none of the terms of $g-h$ collide modulo
  $z^p-1$, $(g-h) \bmod (z^p-1)$ has sparsity exactly $t-\COL{p}$.  It
  follows that $g \bmod (z^p-1)$ has sparsity at most
  $t-\COL{p}+\COL{p}/2=t-\COL{p}/2$.  That is, $s \leq t-\COL{p}/2$,
  and so $\COL{p} \leq 2(t-s)$.  \qed
\end{proof}

\begin{Cor}\label{cor:choose}
  Suppose $g$ has sparsity $t$, $g \bmod (z^q-1)$ has sparsity $s_q$, and 
  $g \bmod (z^p-1)$
  has sparsity $s_p \geq s_q$.  Then $\COL{p} \leq 2\COL{q}$.
\end{Cor}
\begin{proof}\quad\vspace*{-\baselineskip}
  \[
  \begin{array}{rll}
    \COL{p} & \leq 2(t-s_p)~~ & \emph{by the second inequality 
      of Lemma \ref{lem:chooseMostTerms},} \\
    & \leq 2(t-s_q) & \emph{since $s_p \geq s_q$,} \\
    & \leq 2\COL{q} & \emph{by the first inequality of Lemma
      \ref{lem:chooseMostTerms}.} \hspace*{20pt}\qed
  \end{array}
  \]
\end{proof}

Suppose then that we have computed $g \bmod (z^p-1)$, for $p$
belonging to some set of primes $S$, and the minimum value of
$\COL{p}$, $p \in S$, is less than $\gamma$.  Then a prime $p^* \in S$
for which $g \bmod (z^{p^*}-1)$ has maximally many terms satisfies
$\COL{p^*} < 2\gamma$.  We will call such a prime $p^*$ an \emph{ok
  prime}.

We then choose $\gamma = wT$ for an appropriate proportion $w \in
(0,1)$.  We show that setting $w = 3/16$ allows that each recursive
call reduces the sparsity of the subsequent polynomial by at least
half.  This would make $\lambda = \lceil \tfrac{10}{3w}(T-1)\ln(D)
\rceil = \lceil \tfrac{160}{9}(T-1)\ln(D) \rceil$.  As per Lemma
\ref{lem:lambda}, in order to guarantee with probability
$1-\varepsilon$ that we have come across a prime $p$ such that
$\COL{p} \leq \gamma$, we will need to perform $\lceil \log
1/\varepsilon \rceil$ probes of degree $\mathcal{O}(T\log D)$.
Procedure {\tt\ref{proc:FindOkPrime}} summarizes how we find an ok
prime.

\begin{procedure}
\caption{FindOkPrime($\SLP_f, f^*, T_g, D, \varepsilon$)}
\label{proc:FindOkPrime}
\KwIn{\begin{itemize}[noitemsep,nolistsep]
\item $\SLP_f$, a straight-line program that computes a polynomial $f$
\item $f^*$, a current approximation to $f$
\item $T_g$ and $D$, bounds on the sparsity and degree of $g=f-f^*$ 
respectively
\item $\varepsilon$, a bound on the probability of failure
\end{itemize}}
\KwOut{With probability at least $1-\varepsilon$, we return an ``ok prime'' 
for $g=f-f^*$}\vspace{\baselineskip}

$\lambda \longleftarrow \max\left( 21, \left\lceil \tfrac{160}{9}(T_g-1)\ln D 
\right\rceil\right)$\\
$\left({\tt max\_sparsity}, p\right) \longleftarrow (0,0)$\\

\For{$i \longleftarrow 1$ \KwTo $\lceil \log 1/\varepsilon \rceil$ }{
	$p' \longleftarrow $ a random prime in $[ \lambda, 2\lambda]$\\
	\If{ \# of terms of $(f - f^*) \bmod (z^{p'}-1) \geq {\tt max\_sparsity}$}{
		${\tt max\_sparsity} \longleftarrow $ \# of terms of $(f - f^*) \bmod 
		(z^{p'}-1)$\\
		$p \longleftarrow p'$
	}
}
\Return{$p$}
\end{procedure}

A practical application would probably choose random primes by selecting random integer values in $[\lambda, 2\lambda]$ and then applying probabilistic primality testing.  In order to ensure deterministic worst-case run-time, we could pick random primes in the range $[\lambda, 2\lambda]$ by using a sieve method to pre-compute all the primes up to $2\lambda$.


\subsection{Generating an approximation $f^{**}$ of
  $g$}\label{sec:fstarstar}

We suppose now that we have, with probability at least $1 -
\varepsilon$, an ok prime $p$; i.e., a prime $p$ such that $\COL{p}
\leq 2wT$ for a suitable proportion $w$.  We now use this ok prime $p$
to construct a polynomial $f^{**}$ containing {\em most} of the terms
of $g=f-f^*$.

For a set of coprime moduli $\mathcal{Q} = \{q_1, \dots, q_k\}$
satisfying $\prod_{i=1}^k q_i > D$, we will compute $g \bmod
(z^{pq_i}-1)$ for $1 \leq i \leq k$.  Here we make no requirement that
the $q_i$ be prime.  We merely require that the $q_i$ are pairwise
co-prime.

We choose the $q_i$ as follows: denoting the $i\th$ prime by $p_i$, we
set $q_i = p_i^{\lfloor \log_{p_i} x \rfloor}$, for an appropriate
choice of $x$.  That is, we let $q_i$ be the greatest power of the
$i\th$ prime that is no more than $x$.  For $p_i \leq x$, we have $q_i
\geq x/p_i$ and $q_i \geq p_i$.  Either $x/p_i$ or $p_i$ is at least
$\sqrt{x}$, and so $q_i\ge \sqrt{x}$ as well.

By Corollary 1 of Theorem 2 in \cite{RosSch62}, there are more than
$x/\ln x$ primes less than or equal to $x$ for $x \geq 17$.  Therefore
\begin{equation*}
  \prod_{p_i \leq x}q_i \geq \left(\sqrt{x}\right)^{x/\ln x}.
\end{equation*}
As we want this product to exceed $D$, it suffices that
\begin{align*}
  \ln D &< \ln\left( \left(\sqrt{x}\right) ^{x/\ln x}\right) = x/2.
\end{align*}
Thus, if we choose $x \geq \max( 2\ln(D), 17)$ and $k = \lceil x/\ln x
\rceil$, then $\prod_{i=1}^k q_i$ will exceed $D$.  This means $q_i
\in \O(\log D)$ and $pq_i \in \O( T\log^2 D)$.  The number of probes
in this step is $k \in \O( \log(D)/\log\log(D))$.  Since we will use
the same set of moduli $\mathcal{Q} = \{q_1, \dots, q_k\}$ in every
recursive call, we can pre-compute $\mathcal{Q}$ prior to the first
recursive call.


We now describe how to use the images $g \bmod (z^{pq_i}-1)$ to
construct a polynomial $f^{**}$ such that $g-f^{**}$ is at most
$T/2$-sparse.

If $cz^e$ is a term of $g$ that does not collide with any other terms
modulo $z^p-1$, then it certainly will not collide with other terms
modulo $z^{pq}-1$ for any natural number $q$.  Similarly, if
$c^*z^{{e^*} \bmod p}$ appears in $g \bmod (z^p-1)$ and there exists a
unique term $c^*z^{{e^*} \bmod pq_i}$ appearing in $g \bmod
(z^{pq_i}-1)$ for $i=1, 2, \dots, k$, then $c^*z^{e^*}$ is potentially
a term of $g$.  Note that $c^*z^{e^*}$ is not \emph{necessarily} a
term of $g$: consider the following example.

\begin{Exam}
  Let
  \begin{equation*}
    g(z) = 1+z+z^2+z^3 + z^{11+4} - z^{14\cdot 11+4} - z^{15\cdot 11 + 4},
  \end{equation*}
  with hard sparsity bound $T_g=7$ and degree bound $D=170$ and let
  $p=11$.  We have
  \begin{equation*}
    g(z) \bmod (z^{11}-1) = 1+z+z^2+z^3-z^4.
  \end{equation*}
  As $\deg(g)=170<2\cdot 3\cdot 5\cdot 7=210$, it suffices to make the
  probes $g \bmod z^{11q}-1$ for $q=2,3,5,7$.  Probing our remainder
  black-box polynomial, we have
  \begin{align*}
    g \bmod (z^{22}-1) &= 1+z+z^2+z^3-z^{15},\\
    g \bmod (z^{33}-1) &= 1+z+z^2+z^3-z^{26},\\
    g \bmod (z^{55}-1) &= 1+z+z^2+z^3-z^{48},\\
    g \bmod (z^{77}-1) &= 1+z+z^2+z^3-z^{15}.
  \end{align*}
  In each of the images $g \bmod z^{pq}-1$, there is a unique term
  whose degree is congruent to one of $e=0,1,2,3,4$ modulo $p$.  Four
  of these terms correspond to the terms $1,z,z^2,z^3$ appearing in
  $g$.  Whereas the remaining term has degree $e$ satisfying $e = 1
  \bmod 2$, $e = 2 \bmod 3$, $e = 3 \bmod 5$, and $e = 1 \bmod 7$.  By
  Chinese remaindering on the exponents, this gives a term $-z^{113}$
  not appearing in $g$.
\end{Exam}

\begin{procedure}
  \caption{ConstructApproximation($\SLP_f, f^*, D, p,
    \mathcal{Q}$)}\label{proc:ConstructApproximation}
  \KwIn{
    \begin{itemize}[noitemsep,nolistsep]
    \item $\SLP_f$, a straight-line program that computes a polynomial
      $f$
    \item $f^*$, a current approximation to $f$
    \item $D$ a bound on the degree of $g=f-f^*$
    \item $p$, an ok prime for $g$ (with high probability)
    \item $\mathcal{Q}$, a set of co-prime moduli whose product
      exceeds $D$
    \end{itemize}
  } \KwOut{A polynomial $f^{**}$ such that, if $p$ is an ok prime,
    $g-f^{**}$ has sparsity at most $\lfloor T_g/2
    \rfloor$, where $g$ has at most $T_g$ terms.}\vspace{\baselineskip}

  \tcp{Collect images of $g$}
  $\mathcal{E} \longleftarrow $ set of exponents of terms in $(f-f^*) \bmod (z^p-1)$\\
  \For{$q \in \mathcal{Q}$}{
    $h \longleftarrow (f-f^*) \bmod (z^{pq}-1)$\\
    \For{ each term $cz^e$ in $h$ }{ 
      \lIf{ $E_{(e \bmod p), q}$ is already initialized}{ $\mathcal{E}
        \longleftarrow \mathcal{E}/\{e \bmod p\}$} 
      \lElse{ $E_{(e \bmod p), q} \longleftarrow e \bmod q$ } 
    } } \vspace{\baselineskip}

  \tcp{Construct terms of new approximation of $g$, $f^{**}$}
  $f^{**} \longleftarrow 0$\\
  \For{ $e_p \in \mathcal{E}$ }{
    $e \longleftarrow $ least nonnegative solution to $\{e = E_{e_p, q} \bmod q $ $|$ $q \in \mathcal{Q}\}$\\
    $c \longleftarrow $ coefficient of $z^{e_p}$ term in $(f-f^*) \bmod (z^p-1)$\\
    \lIf{$e \leq D$}{ $f^{**} \longleftarrow f^{**} + cz^e$ } }
  \Return{ $f^{**}$ }

\end{procedure}

\vspace*{-3pt}
\begin{Defn}
  Let $c^*z^{e^*}$, $e^* \leq D$ be a monomial such that $c^*z^{e^*
    \bmod p}$ appears in $g \bmod z^p-1$, and $c^*z^{e^* \bmod pq_i}$
  is the unique term of degree congruent to $e^*$ modulo $p$ appearing
  in $g \bmod (z^{pq_i}-1)$ for each modulus $q_i$.  If $c^*z^{e^*}$
  is not a term of $g$ we call it a {\em deceptive term}.
\end{Defn}

Fortunately, we can detect a collision comprised of only two terms.
Namely, if $c_1z^{e_1} + c_2z^{e_2}$ collide, there will exist a $q_i$
such that $q_i \nmid (e_1-e_2)$.  That is, $g \bmod (z^{pq_i}-1)$ will
have two terms whose degree is congruent to $e_1 \bmod p$.  Once we
observe that, we know the term $(c_1+c_2)z^{e_1 \bmod p}$ appearing in
$g \bmod (z^p-1)$ was not a distinct term, and we can ignore exponents
of the congruence class $e_1 \bmod p$ in subsequent images $g \bmod
(z^{pq_j}-1)$.

Thus, supposing $g \bmod (z^p-1)$ has at most $2\gamma$ colliding
terms and at least $t-2\gamma$ non-colliding terms, $f^{**}$ will have
the $t-2\gamma$ non-colliding terms of $g$, plus potentially an
additional $\tfrac{2}{3}\gamma$ deceptive terms produced by the
colliding terms of $g$.  In any case, $g-f^{**}$ has sparsity at most
$\tfrac{8}{3}\gamma$.  Choosing $\gamma = \tfrac{3}{16}T_g$ guarantees
that $g-f^{**}$ has sparsity at most $T_g/2$.  This would make
$\lambda = \lceil \tfrac{160}{9}(T_g-1)\ln(D) \rceil$.

Procedure {\tt \ref{proc:ConstructApproximation}} gives a pseudocode
description of how we construct $f^{**}$.

If we find a prospective term in our new approximation $f^{**}$ has
degree greater than $D$, then we know that term must have been a
deceptive term and discard it.  There are other obvious things we can
do to recognize deceptive terms which we exclude here.  For instance,
we should check that all terms from images modulo $z^{pq}-1$ whose
degrees agree modulo $p$ share the same coefficient.


\subsection{Recursively interpolating $f-f^*$}

\begin{procedure}
\caption{Interpolate($\SLP_f, T, D, \mu$) }\label{proc:Interpolate}
\KwIn{
	\begin{itemize}[noitemsep,nolistsep]
	\item $\SLP_f$, a straight-line program that computes a polynomial $f$
	\item $T$ and $D$, bounds on the sparsity and degree of $f$, respectively
	\item $\mu$, an upper bound on the probability of failure
	\end{itemize}
}
\KwOut{
	With probability at least $1-\mu$, we return $f$
}\vspace{\baselineskip}

$x \longleftarrow \max( 2\ln(D), 17)$\\
$\mathcal{Q} \longleftarrow \{ p^{\lfloor \log_p x \rfloor} : p \text{ is prime}, p \leq x \}$\\
\Return{ {\tt\ref{proc:InterpolateRecurse}$(\SLP_f, 0, T, D, \mathcal{Q}, \mu/(\log T + 1) )$ }}
\end{procedure}

\begin{procedure}
\caption{ InterpolateRecurse($\SLP_f, f^*, T_g, D, \mathcal{Q}, \varepsilon$ ) }\label{proc:InterpolateRecurse}
\KwIn{
	\begin{itemize}[noitemsep,nolistsep]
	\item $\SLP_f$, a straight-line program that computes a polynomial $f$
	\item $f^*$, a current approximation to $f$
	\item $T_g$ and $D$, bounds on the sparsity and degree of $g=f-f^*$, respectively
	\item $\mathcal{Q}$, a set of coprime moduli whose product is at least $D$
	\item $\varepsilon$, a bound on the probability of failure at one recursive step
	\end{itemize}
}
\KwOut{With probability at least $1 -\mu$, the algorithm outputs $f$}\vspace{\baselineskip}

\lIf{ $T_g = 0$ }{ \Return{ $f^*$ } } 

$p \longleftarrow {\tt\ref{proc:FindOkPrime}}(\SLP_f, f^*, T_g, D, \varepsilon)$\\
$f^{**} \longleftarrow {\tt\ref{proc:ConstructApproximation}}(\SLP_f, f^*, D, 
p, \mathcal{Q})$\\

\Return{ {\tt\ref{proc:InterpolateRecurse}}$(\SLP_f, f^*+f^{**}, \lfloor T_g/2 \rfloor, D, \mathcal{Q}, \varepsilon ) $}
\end{procedure}

Once we have constructed $f^{**}$, we refine our approximation $f^*$
by adding $f^{**}$ to it, giving us a new difference $g=f-f^*$
containing at most half the terms of the previous polynomial $g$.  We
recursively interpolate our new polynomial $g$.  With an updated
sparsity bound $\lfloor T_g/2 \rfloor$, we update the values of
$\gamma$ and $\lambda$ and perform the steps of Sections \ref{sec:ok}
and \ref{sec:fstarstar}.  We recurse in this fashion $\log T$
times.  Thus, the total number of probes becomes
\[
\O\left( \log T ( \tfrac{\log D}{\log\log D} + \log(1/\varepsilon)) \right),
\]
of degree at most $\O( T\log^2 D)$.

Note now that in order for this method to work we need that, at every
recursive call, we in fact get a good prime, otherwise our sparsity bound
on the subsequent difference of polynomials could be incorrect.  At
every stage we succeed with probability $1-\varepsilon$, thus the
probability of failure is $1-(1-\varepsilon)^{\lceil \log T \rceil}$.
This is less than $\lceil \log T \rceil \varepsilon$.  If we want to
succeed with probability $\mu$, then we can choose $\varepsilon =
\tfrac{ \mu }{ \log T + 1} \in \O( \tfrac{ \mu }{ \log T
})$.  

{\tt\ref{proc:Interpolate}} pre-computes our set of moduli $\mathcal{Q}$, then makes the first recursive call to {\tt\ref{proc:InterpolateRecurse}}, which subsequently calls itself.  


\subsection{A cost analysis}

We analyse the cost of our algorithm, thereby proving Theorems
\ref{thm:cost} and \ref{thm:probe}.

\subsubsection{Pre-computation.}

Using the wheel sieve \citep{Pri82}, we can compute the set of primes
up to $x \in \O( \log D)$ in $\softO( \log D )$ bit operations. From
this set of primes we obtain $\Q$ by computing $p^{\lfloor
  \log_p x \rfloor}$ for $p \leq \sqrt{x}$ by way of
squaring-and-multiplying. For each such prime, this costs $\softO( \log x
)$ bit operations, so the total cost of computing $\Q$ is $\softO(
\log D)$.

\subsubsection{Finding ok primes. }

In one recursive call, we will look at some $\log 1/\varepsilon =
\O(\log 1/\mu \log\log T)$ primes in the range $[ \lambda, 2\lambda]$
in order to find an ok prime. Any practical implementation would
select such primes by using probabilistic primality testing on random
integer values in the range $[ \lambda, 2\lambda]$; however, the
probabilistic analysis of such an approach, in the context of our
interpolation algorithm, becomes somewhat ungainly. We merely note
here that we could instead pre-compute primes up to our initial value
of $\lambda \in \O(T\log D)$ in $\softO( T\log
D)$ bit operations by way of the wheel sieve.

Each prime $p$ is of order $T\log D$, and so, per our discussion in
Section \ref{sec:intro}, each probe costs $\softO(LT\log D)$ ring
operations and similarly many bit operations. Considering the 
$\O(\log T)$ recursive calls, this totals $\softO( LT\log D\log 1/\mu )$
ring and bit operations.

\subsubsection{Constructing the new approximation $f^{**}$. } 

Constructing $f^{**}$ requires $\softO( \log D )$ probes of degree
$\softO( T\log^2 D )$. This costs $\softO( LT \log^3 D)$ ring and bit
operations. Performing these probes at each $\O( \log T)$ recursive
call introduces an additional factor of $\log T$, which does not
affect the ``soft-Oh'' complexity. This step dominates the cost of the
algorithm.

Building a term $cz^e$ of $f^{**}$ amounts to solving a set of congruences.  
By Theorem 5.8 of \cite{GatGer03}, this requires some $\O( \log^2 D )$ word 
operations. Thus the total cost of Chinese remaindering to construct $f^{**}$ 
becomes $\O( T\log^2 D)$.  Again, the additional $\log T$ factor due to the 
recursive calls does not affect the stated complexity.


\section{Conclusions}

We have presented a recursive algorithm for interpolating a polynomial
$f$ given by a straight-line program, using probes of smaller degree
than in previously known methods. We achieve this by looking for
``ok'' primes which separate most of the terms of $f$, as opposed to
``good'' primes which separate all of the terms of $f$. As is seen in
Table \ref{tab:comp}, our algorithm is an improvement over previous
algorithms for moderate values of $T$.

This work suggests a number of problems for future work. We believe
our algorithms have the potential for good numerical stability, and
could improve on \citeauthor{GieRoc11}'s \citeyearpar{GieRoc11} work
on numerical interpolation of sparse complex polynomials, hopefully
capitalizing on the lower degree probes. Our Monte Carlo algorithms
are now more efficient than the best known algorithms for polynomial
identity testing, and hence these cannot be used to make them error
free. We would ideally like to expedite polynomial identity testing of
straight-line programs, the best known methods currently due to
\cite{BlaHar09}. Finally, we believe there is still room for
improvement in sparse interpolation algorithms. The vector of
exponents of $f$ comprises some $T\log D$ bits. Assuming no
collisions, a degree-$\ell$ probe gives us some $t\log \ell$ bits of
information about these exponents. One might hope, aside from some
seemingly rare degenerate cases, that $\log D$ probes of degree $T\log
D$ should be sufficient to interpolate $f$.


\section{Acknowledgements}

We would like to thank Reinhold Burger and Colton Pauderis for their feedback 
on a draft of this 
paper.

\renewcommand\bibsection{\section*{References}}

\bibliography{rsi}

\begin{thebibliography}{15}
\providecommand{\natexlab}[1]{#1}
\providecommand{\url}[1]{\texttt{#1}}
\expandafter\ifx\csname urlstyle\endcsname\relax
  \providecommand{\doi}[1]{doi: #1}\else
  \providecommand{\doi}{doi: \begingroup \urlstyle{rm}\Url}\fi

\bibitem[Ben-Or and Tiwari(1988)]{BenTiw88}
Michael Ben-Or and Prasoon Tiwari.
\newblock A deterministic algorithm for sparse multivariate polynomial
  interpolation.
\newblock In \emph{Proceedings of the twentieth annual ACM symposium on Theory
  of computing}, pages 301--309. ACM, 1988.

\bibitem[Bl{\"a}ser et~al.(2009)Bl{\"a}ser, Hardt, Lipton, and
  Vishnoi]{BlaHar09}
Markus Bl{\"a}ser, Moritz Hardt, Richard~J. Lipton, and Nisheeth~K. Vishnoi.
\newblock Deterministically testing sparse polynomial identities of unbounded
  degree.
\newblock \emph{Information Processing Letters}, 109\penalty0 (3):\penalty0
  187--192, 2009.

\bibitem[Bruno et~al.(2002)Bruno, Heintz, Matera, and Wachenchauzer]{BruHei02}
Nicolas Bruno, Joos Heintz, Guillermo Matera, and Rosita Wachenchauzer.
\newblock Functional programming concepts and straight-line programs in
  computer algebra.
\newblock \emph{Mathematics and Computers in Simulation}, 60\penalty0
  (6):\penalty0 423--473, 2002.
\newblock \doi{10.1016/S0378-4754(02)00035-6}.

\bibitem[B{\"u}rgisser et~al.(1997)B{\"u}rgisser, Clausen, and
  Shokrollahi]{BCS97}
Peter B{\"u}rgisser, Michael Clausen, and M.~Amin Shokrollahi.
\newblock \emph{Algebraic Complexity Theory}, volume 315 of \emph{Grundlehren
  der mathematischen Wissenschaften}.
\newblock Springer-Verlag, 1997.

\bibitem[Cantor and Kaltofen(1991)]{CanKal91}
David~G. Cantor and Erich Kaltofen.
\newblock On fast multiplication of polynomials over arbitrary algebras.
\newblock \emph{Acta Informatica}, 28:\penalty0 693--701, 1991.

\bibitem[de~Prony(1795)]{Prony95}
R.~de~Prony.
\newblock Essai exp\'{e}rimental et analytique sur les lois de la
  dilabilit\'{e} et sur celles de la force expansive de la vapeur de l'eau et
  de la vapeur de l'alkool, \`{a} diff\'{e}rentes temp\'{e}ratures.
\newblock \emph{J. de l'\'{E}cole Polytechnique}, 1:\penalty0 24--76, 1795.

\bibitem[Garg and Schost(2009)]{GarSch09}
Sanchit Garg and \'{E}ric Schost.
\newblock Interpolation of polynomials given by straight-line programs.
\newblock \emph{Theor. Comput. Sci.}, 410\penalty0 (27-29):\penalty0
  2659--2662, June 2009.
\newblock ISSN 0304-3975.
\newblock \doi{10.1016/j.tcs.2009.03.030}.
\newblock URL \url{http://dx.doi.org/10.1016/j.tcs.2009.03.030}.

\bibitem[Gathen and Gerhard(2003)]{GatGer03}
Joachim von~zur Gathen and Jurgen Gerhard.
\newblock \emph{Modern Computer Algebra}.
\newblock Cambridge University Press, New York, NY, USA, 2nd edition, 2003.
\newblock ISBN 0521826462.

\bibitem[Giesbrecht and Roche(2011)]{GieRoc11}
Mark Giesbrecht and Daniel~S. Roche.
\newblock Diversification improves interpolation.
\newblock \emph{ISSAC '11}, pages 123--130, 2011.
\newblock \doi{10.1145/1993886.1993909}.
\newblock URL \url{http://doi.acm.org/10.1145/1993886.1993909}.

\bibitem[Giesbrecht et~al.(2009)Giesbrecht, Labahn, and Lee]{GieLab09}
Mark Giesbrecht, George Labahn, and Wen-shin Lee.
\newblock Symbolic--numeric sparse interpolation of multivariate polynomials.
\newblock \emph{Journal of Symbolic Computation}, 44\penalty0 (8):\penalty0
  943--959, 2009.

\bibitem[Kaltofen(1989)]{Kal89}
Erich Kaltofen.
\newblock Factorization of polynomials given by straight-line programs.
\newblock In \emph{Randomness and Computation}, pages 375--412. JAI Press,
  1989.

\bibitem[Kaltofen et~al.(1990)Kaltofen, Lakshman, and Wiley]{KalLak90}
Erich Kaltofen, Y.~N. Lakshman, and John-Michael Wiley.
\newblock Modular rational sparse multivariate polynomial interpolation.
\newblock In \emph{Proceedings of the international symposium on Symbolic and
  algebraic computation}, ISSAC '90, pages 135--139, New York, NY, USA, 1990.
  ACM.
\newblock \doi{10.1145/96877.96912}.

\bibitem[Pritchard(1982)]{Pri82}
Paul Pritchard.
\newblock Explaining the wheel sieve.
\newblock \emph{Acta Informatica}, 17\penalty0 (4):\penalty0 477--485, 1982.

\bibitem[Rosser and Schoenfeld(1962)]{RosSch62}
J.~Barkley Rosser and Lowell Schoenfeld.
\newblock Approximate formulas for some functions of prime numbers.
\newblock \emph{Illinois J. Math.}, 6:\penalty0 64--94, 1962.
\newblock ISSN 0019-2082.

\bibitem[Sturtivant and Zhang(1990)]{StuZha90}
Carl Sturtivant and Zhi-Li Zhang.
\newblock Efficiently inverting bijections given by straight line programs.
\newblock In \emph{Foundations of Computer Science, 1990. Proceedings., 31st
  Annual Symposium on}, pages 327--334. IEEE, Oct 1990.
\newblock \doi{10.1109/FSCS.1990.89551}.

\end{thebibliography}

\end{document}